\documentclass[11pt,letterpaper]{article}
\usepackage[tmargin=1in,bmargin=1in,lmargin=1in,rmargin=1in]{geometry}
\usepackage{setspace}
\usepackage{authblk}
\usepackage{gensymb}
\usepackage{diagbox}
\usepackage{xcolor}
\usepackage{epsfig}
\usepackage{cite}
\usepackage{graphicx}
\usepackage{amsmath, mathtools}
\usepackage{ragged2e}
\usepackage{amsthm}
\usepackage{amssymb}
\usepackage{tabularx}
\usepackage{wrapfig}  % enables floating wrapped figures
\usepackage{url}
\usepackage{hyperref}
\usepackage{subcaption}
\usepackage{booktabs}
\theoremstyle{definition}
\newtheorem{definition}{Definition}
\theoremstyle{lemma}
\newtheorem{lemma}{Lemma}
\title{Multistage Economic MPC for Systems with a Cyclic Steady State: A Gas Network Case Study}
\author[1]{Sakshi Naik}
\author[2]{Lavinia Ghilardi}
\author[3] {Robert Parker}
\author[1]{Lorenz T. Biegler}
\affil[1]{Carnegie Mellon University, Pittsburgh, PA 15213, USA}
\affil[2]{Politecnico di Milano, Department of Energy, Milano, 20154, Italy}
\affil[3]{Applied Math and Plasma Physics Group, Los Alamos National Laboratory, Los Alamos, NM, USA
}
\date{February 2025}

\begin{document}

\maketitle
\section*{Abstract}
Multistage model predictive control (MPC) provides a robust control strategy for dynamic systems with uncertainties and a setpoint tracking objective. Moreover, extending MPC to minimize an economic cost instead of tracking a pre-calculated optimal setpoint improves controller performance. In this paper, we develop a formulation for multistage economic MPC which directly minimizes an economic objective function. The multistage economic MPC framework is extended for systems with a cyclic steady state (CSS) and stability is guaranteed by employing a Lyapunov-based stability constraint. The multistage economic MPC framework is validated on two natural gas network case studies to minimize the net energy consumption during gas transmission. In both instances, the multistage economic MPC effectively manages uncertain demands by preventing constraint violations and guides the network to its optimal cyclic operating conditions. The Lyapunov function remains bounded in both instances, validating the robust stability of the controller. 
\section{Introduction}

Model predictive control (MPC) is a moving horizon controller based on optimal decisions for a dynamic process. The MPC recursively solves a constrained optimization problem while updating its initial state with the current state of the plant. The optimal control inputs suggested by the MPC are fed back into the plant to calculate its future state. Nonlinear MPC (NMPC) has been widely studied on industrial applications \cite{franke_integration_2007, nagy_real-time_2007, parker_dynamic_2022} where the controller optimizes a nonlinear dynamic optimization problem at every step to track an optimal setpoint trajectory. Multilayer control strategies (such as NMPC in combination with a real-time optimizer) have been shown to work better by updating the setpoint \cite{de_prada_integration_2017, patron_real-time_2020}; however, they require greater computational effort. An alternative control approach is the economic NMPC (E-NMPC) where an economic objective function is directly minimized/maximized instead of tracking a precalculated setpoint trajectory \cite{ rawlings_fundamentals_2012, skogestad_advanced_2023}. This approach enables the direct optimization of process performance metrics, such as energy efficiency or operational costs, over a finite prediction horizon. 

A nominal MPC controller relies on a good process model with predetermined parameters. However, an ideal controller should be robust to model uncertainty and disturbances in the plant. 
This is especially important for constraint satisfaction and improved performance.
Stochastic as well as robust approaches have been proposed to handle these issues under uncertainty \cite{grossmann_recent_2016}. Moreover, min-max MPC provides robustness by minimizing the worst-case cost \cite{campo_robust_1987, raimondo_min-max_2009}. However, this can lead to overly conservative solutions, as it always minimizes the worst-case cost without considering uncertainty realization in the future. An alternative less conservative approach is the tube-based MPC \cite{mayne_tube-based_2011} where a nominal trajectory is calculated and the MPC forces the disturbed system to lie within a tube that is centered around the nominal trajectory. While the tube-based MPC guarantees stability and constraint satisfaction in the presence of uncertainties, it does not guarantee optimal performance. MPC with back-off constraints has also received attention for its simple and effective implementation \cite{krog_simple_2024}. The main idea is to calculate backoff constraints using Monte Carlo simulations until there are no constraint violations. However, both tube-based MPC and MPC with backoff constraints rely on bound tightening strategies that can lead to overly conservative solutions.

A robust control strategy was introduced by Lucia et al. \cite{lucia_multi-stage_2013} where the evolution of uncertainty is represented using a scenario tree that contains recourse variables for the controls. They show that this multistage NMPC is less conservative and outperforms a nominal min-max NMPC in presence of uncertainty. Zhou and Biegler \cite{zhou_optimal_2021} extended the multistage NMPC framework to a parallelizable advanced-step multistage NMPC. Mdoe et al. \cite{mdoe_adaptive_2021} proposed a computationally efficient adaptive horizon multistage NMPC, where the length of the control horizon is determined using NLP sensitivities. Several other advancements and applications for multistage NMPC can be found here \cite{lucia_potential_2015, thangavel_robust_2020, lin_multistage_2022, lucia_stability_2020}. There is limited literature on multistage E-NMPC \cite{tatulea-codrean_multi-stage_2020, naik_multistage_2023} due to challenges with controller stability when the stage costs are economic functions instead of quadratic tracking terms. 

In this paper, we develop a formulation for multistage E-NMPC for systems with cyclic steady-state. Controller stability is ensured by explicitly enforcing the Lyapunov descent property. We show through theory and numerical experiments that the controller has input to state practical stability. The multistage E-NMPC is demonstrated in two case studies of natural gas pipeline networks with uncertain demands. 

\section{E-NMPC formulation}
An E-NMPC controller solves a constrained optimization problem, directly minimizing an economic objective function. The plant model is then simulated using the first control input obtained from the optimization to calculate the new state of the plant. Then, the initial state of the E-NMPC controller is updated with the new plant state and this process is repeated. Section (\ref{sec:std-enmpc-formulation}) formulates the nominal E-NMPC optimization problem assuming no uncertainty for a system with cyclic steady state. 
\subsection{Nominal E-NMPC}\label{sec:std-enmpc-formulation}
Systems with a cyclic steady state have periodic operation that repeats after each cycle. Let $K$ be the length of the cycle and $N$ be the number of times the cycle repeats in the control horizon. The length of the control horizon is therefore given by $NK$.
The optimization formulation for an E-NMPC employed on a system with cyclic steady states is developed below:
\begin{equation}\label{sec:enmpc-formulation}
    \begin{aligned}
\min \quad & \sum_{i = 0}^{NK-1}l^{ec}(z_i, v_i) \\
\text{s.t.} \quad  & z_{i+1} = f(z_i, v_i)\quad  \forall i \in \{0, .. NK-1\} \\
& z_0 = x_k \\
& z_{(N-1)K} = z_{0}^*  \\
& v_{(N-1)K + i} = v_i^*, \quad \forall i \in \{1, ..K-1\}\\
& V_k - V_{k-1} \leq -\delta l^{tr}(x_{k-1}, u_{k-1}) , \quad \delta \in (0,1]\\
& z_i \in \mathcal{X}, v_i \in \mathcal{U}, \quad \forall i \in \{0, ... NK-1\}
\end{aligned}
\end{equation}
The states and the controls in the E-NMPC are denoted by $z$ and $v$ respectively while the states and the controls in the plant are denoted by $x$ and $u$ respectively. $z^*$, $v^*$ is the optimal CSS determined from (\ref{sec:css-formulation}). The different notations for the variables in the controller and the plant are important to account for any mismatch in the plant and the controller model. 
The E-NMPC minimizes the sum of economic stage costs $l^{ec}$ over the control horizon by optimizing the control variables. 
The current state of the plant at time $k$, denoted by $x_k$, is loaded into the controller at the initial time point of the controller $(i = 0)$. 
An E-NMPC is not necessarily stable since the objective function is not a $\mathcal{K}$ function. Therefore, special measures are necessary to ensure the stability of the E-NMPC. In this case, we employ a Lyapunov-based stability constraint to ensure that the system is stable.
The Lyapunov function $V_k$ for $x_k$ is given by:
\begin{equation}\label{eq: lyapunov-definition}
    V_k = V(x_k) = \sum_{i = 0}^{NK-1}l^{tr}(z_{i|k}, v_{i|k})
\end{equation}
and the tracking cost is defined as:
\begin{equation*}
    l^{tr}(z_i, v_i) = (z_i - z_i^*)^2 + (v_i - v_i^*)^2
\end{equation*}
The stability of a nominal E-NMPC is ensured by enforcing the tracking Lyapunov function to strictly decrease \cite{griffith_advances_nodate}. The parameter $\delta \in (0,1]$ controls the rate of descent of the Lyapunov function.
The cyclic steady-state constraints on the states and the controls drive the controller to the cyclic steady state in the last period. The optimal cyclic steady state $(z_i^*, v_i^*)$ is pre-computed by solving the following optimization problem \cite{huang_stability_2011}:

\begin{equation}\label{sec:css-formulation}
    \begin{aligned}
\min \quad & \sum_{i = 0}^{K-1}l^{ec}(z_i, v_i) \\
\text{s.t.} \quad  & z_{i+1} = f(z_i, v_i)\quad  \forall i \in \{0, .. K-1\} \\
& z_0 = z_K \\
& z_i \in \mathcal{X}, v \in \mathcal{U}, \quad \forall i \in \{0, ... K-1\}
\end{aligned}
\end{equation}
The nominal E-NMPC formulation does not account for uncertainties in the model parameters. In the presence of uncertainty in the model parameters, a nominal E-NMPC described in equation set (\ref{sec:enmpc-formulation}) is not enough to satisfy all the constraints. The performance of a nominal E-NMPC in the presence of uncertain parameters in the plant is shown in sections (\ref{sec:non-vs-multistage-test-network}, \ref{sec: std-enmpc-unc-demands}).

\subsection{Multistage Economic NMPC}\label{sec:multistage-enmpc}
Lucia et al. \cite{lucia_multi-stage_2013} introduced the multistage NMPC as a robust control strategy in presence of uncertainty. A multistage NMPC explicitly accounts for uncertainties in model parameters and their evolution in the time horizon. Two main assumptions are made in the multistage NMPC framework:
\begin{enumerate}
    \item The uncertain parameters follow a discrete probability distribution with a finite set of parameter values
    \item The uncertain parameters are exogenous and the uncertainty is resolved before the next time step. 
\end{enumerate}

 A multistage E-NMPC minimizes a weighted expected cost function across multiple uncertain scenarios to make robust decisions. Following the approach in \cite{lucia_multi-stage_2013}, Figure (\ref{fig:multistage-tree}) shows a typical fully-branched scenario tree in a multistage E-NMPC with an uncertain parameter with three discrete realizations of the uncertain parameter. $z_i^c$, $v_i^c$, $w_i^c$ represent the state, controls and uncertain parameter realization at time $i$ in a scenario $c$.
\begin{figure}[h]
    \centering
    \includegraphics[width=0.8
\linewidth]{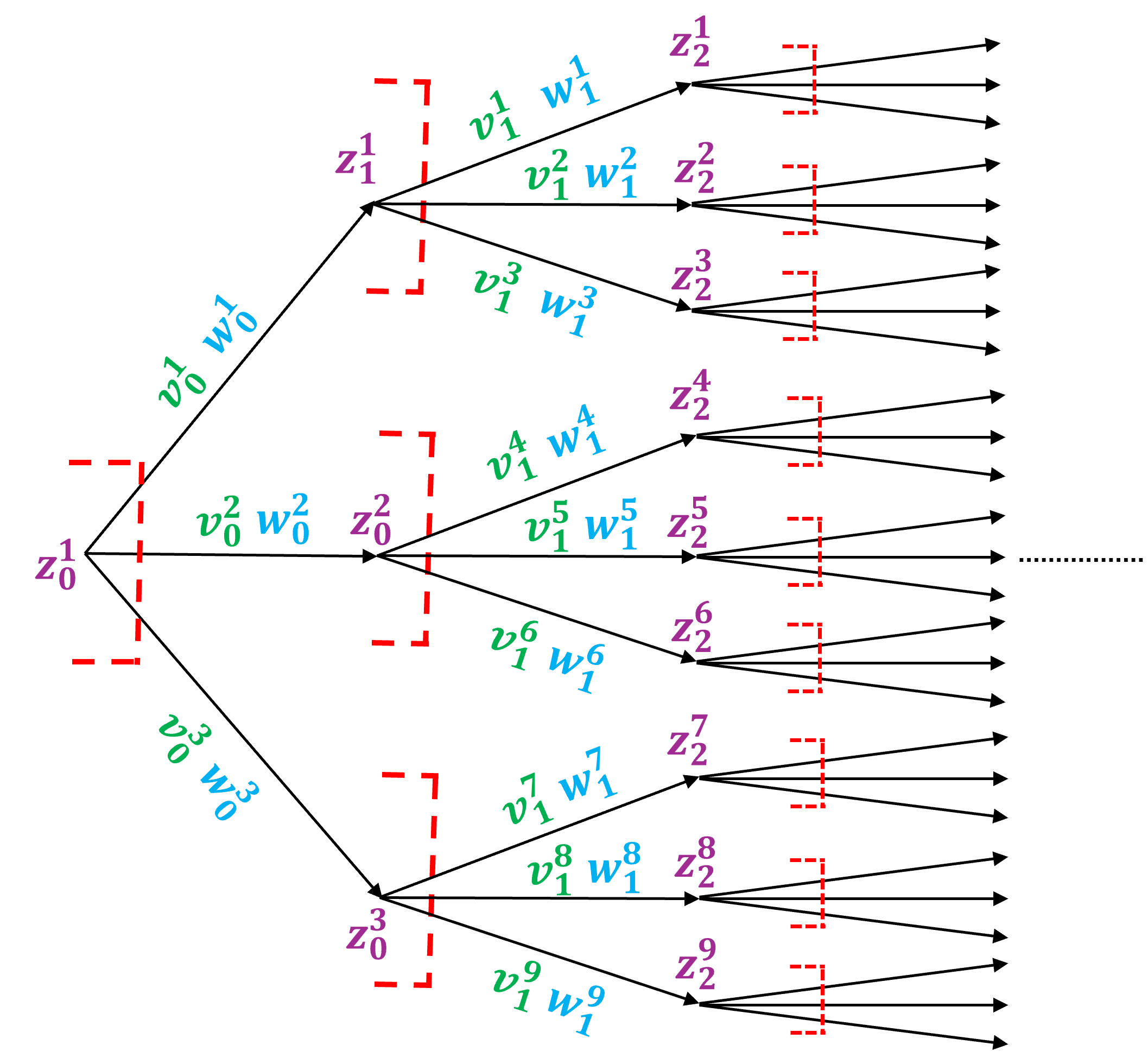}
    \caption{Fully branched multistage NMPC controller with states $z$, controls $v$ and uncertain parameter $w$. The red dotted lines indicate non-anticipativity constraints on the control variables}
    \label{fig:multistage-tree}
\end{figure}
A ``scenario'' is defined as the collection of all uncertainty realizations over a particular controller horizon, starting from the root node i.e. the current state of the plant to a leaf node. For example in Figure  (\ref{fig:multistage-tree}), scenario 1 is given by $\{w_0^1, w_1^1, w_2^1, ..., w_{H-1}^1\}$, scenario 2 is given by $\{w_0^1, w_1^1, w_2^1, ..., w_{H-1}^2\}$ and so on. For a controller with horizon $H$, and $\mathcal{|M|}$ discrete samples (e.g. maximum, nominal and minimum) of the uncertain parameters in each period, the number of scenarios grows as $\mathcal{|M|}^{H}$ in the multistage controller. 

Since the value of the uncertain parameter $w_i$ is not realized in advance, the control inputs $v_i^c$ must be equal in all edges that emanate from a common state $z_{i}^c$. This condition is enforced in the controller by the non-anticipativity constraints in the multistage controller. For example, in Figure (\ref{fig:multistage-tree}), $v_1^1$ and $v_1^2$ must be the same but $v_1^1$ and $v_1^4$ can be different. As the number of scenarios increase and the controller horizon becomes longer, the problem becomes intractable due to increase in problem size and non-anticipativity constraints. 

Therefore, a multistage controller with a robust horizon is proposed in the literature \cite{lucia_multi-stage_2013}. It assumes that the uncertain parameters only need to be handled within the length of the robust horizon while they can be considered to be at a constant level of uncertainty beyond the robust horizon. Since controller only implements the first control action in the plant before it re-optimizes, the robust horizon approach is sufficient to handle uncertainty without solving the full scenario tree. Figure (\ref{fig:multistage-robust-horizon}) shows a multistage E-NMPC with a robust horizon of 1. For a controller with a robust horizon $H_r$, and $|\mathcal{M}|$ discrete realizations of the uncertain parameter, the number of scenarios grows as $|\mathcal{M}|^{H_r}$. 
This is a significant reduction in the number of scenarios and non-anticipativity constraints compared to the fully branched multistage E-NMPC when $H_r << H$. 
In the past, researchers have considered multistage NMPC with short robust horizons of length 1 and 2 \cite{lucia_multi-stage_2013, yu_advanced-step_2019, naik_multistage_2023, lin_multistage_2022} leading to a fast tractable controller with very good results. A mathematical formulation for the multistage E-NMPC with a robust horizon is given in (\ref{sec:multistage-formulation}). 
\begin{figure}[h]
    \centering
    \includegraphics[width=0.7
\linewidth]{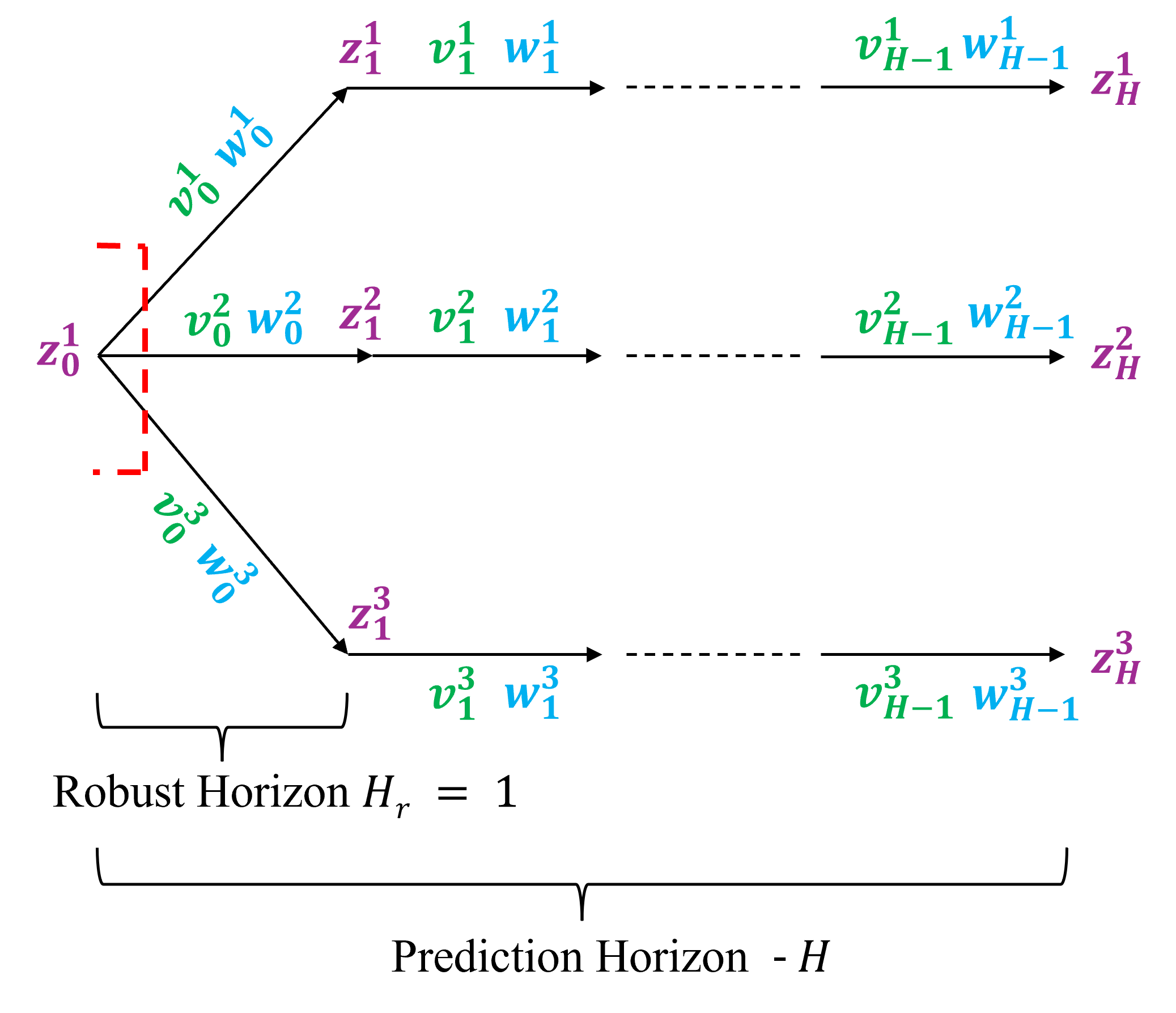}
    \caption{Multistage NMPC controller with a robust horizon of length 1. The red dotted lines indicate non-anticipativity constraints on the control variables within the robust horizon}
    \label{fig:multistage-robust-horizon}
\end{figure}
\begin{equation}\label{sec:multistage-formulation}
    \begin{aligned}
\min \quad &\mathbb{E}_{\mathbb{C}}\Bigg[\sum_{i = 0}^{H-1}l^{ec}(z_i^c, v_i^c, w_i^c)\Bigg]  + \lambda \xi\\
\text{s.t.} \quad  & z_{i+1}^c = f(z_i^c, v_i^c, w_i^c)\quad  \forall i \in \{0, .. H-1\}, \forall c \in \mathbb{C} \\
& v_i^c = v_i^{\bar{c}}, \textit{ if } z_{i}^c = z_{i}^{\bar{c}} \quad \forall i \in \{1, .., H_r\}, \forall c, \bar{c} \in \mathbb{C}\\
& w_i^c = w_i^{\bar{c}} \quad \forall i \in \{H_r+1, ...,H-1\}, \forall c, \bar{c} \in \mathbb{C} \\
& z_0 = x_k \\
& z_{(N-1)K}^c = z_{0}^{c*}  \quad \forall c \in \mathbb{C}\\
& v_{(N-1)K + i}^c = v_i^{c*}\quad \forall i \in \{1, ..K-1\},  \forall c \in \mathbb{C}\\
& \mathbb{E}_\mathbb{C}[(V_k^c-V_{k-1}^c)]\leq -\delta \mathbb{E}_\mathbb{C}[l^{tr}(x_{k-1}^c, u_{k-1}^c)]  + \xi, \quad \delta \in (0,1]\\
& \xi \geq 0 \\
& z_i^c \in \mathcal{X}, v_i^c \in \mathcal{U}, \quad \forall i \in \{0, ... H-1\},  \forall c \in \mathbb{C}
\end{aligned}
\end{equation}
$\mathbb{C}$ is the set of all scenarios in the multistage E-NMPC. The expected value over the set of all scenarios $\mathbb{E}_\mathbb{C}(.)$ is defined as follows: 
\begin{equation}
    \mathbb{E}_\mathbb{C}\Bigg[\sum_{i = 0}^{H-1}l^{ec}(z_i^c, v_i^c, w_i^c)\Bigg] = \sum_{c\in \mathbb{C}}p^c\sum_{i = 0}^{H-1}l^{ec}(z_i^c, v_i^c, w_i^c)
\end{equation}
\begin{equation}\label{eq: expected-lyapunov}
    \mathbb{E}_\mathbb{C}[V_k^c] = \sum_{c\in \mathbb{C}}p^cV_k^c
\end{equation}
\begin{equation}
    \mathbb{E}_\mathbb{C}[l^{tr}(x_{k-1}^c, u_{k-1}^c)] = \sum_{c\in \mathbb{C}}p^cl^{tr}(x_{k-1}^c, u_{k-1}^c)
\end{equation}
$p^c$ is the probability of occurrence of scenario $c$. In this case, each scenario is assigned an equal probability of occurence. The stability constraint from formulation (\ref{sec:enmpc-formulation}) is modified to bound the expected value of the Lyapunov function over the different scenarios \cite{yu_advanced-step_2019}. $V_k^c$ is the Lyapunov function of scenario $c$ and $\xi$ is the slack variable in the stability constraint. The slack variable $\xi$, weighted by $\lambda$, that relaxes the descent condition is also minimized in the objective function to maintain it to be close to zero whenever possible. 
% It is known that at each state $z_{i}^c$, the uncertain parameter $w_{i}$ an take one of the $\mathcal{M}$ uncertainty realizations. Each branch that has sis weighed equally assuming there is equal probability of each uncertainty realization.
% Therefore the weights $p_i^c$ satisfy:
% \begin{equation}
%     \sum_{j \in \mathcal{M}}p_i^j = 1, \quad \forall i \in \{0 , .. H-1\}
% \end{equation}
% For each scenario, the weight $p^c$ is given by a product of the weights of the branches that constitute the scenario.
% \begin{equation}
%     p^c = \prod_{i = 0}^{H-1}p_i^{j_c} \quad  \forall c \in \mathbb{C}
% \end{equation}
% Where $j_c$ consists of all branches that constitute the scenario $c$ and $\mathbb{C}$ is the set of all scenarios. 
The cyclic steady state constraints are written separately for each scenario of the multistage E-NMPC to drive each branch of the multi stage E-NMPC to it's own optimal CSS. An offline calculation is done to calculate the optimal cyclic steady state for each uncertain scenario. 

\section{Stability Properties}
This section establishes the Input to State practical Stability (ISpS) for the multistage E-NMPC. Many definitions are taken from \cite{lucia_stability_2020, yu_advanced-step_2019} where the authors prove ISpS stability of the multistage-NMPC with tracking objective, as many of these concepts directly apply to the multistage-NMPC with economic objective.
The open loop dynamic system to be controlled is given by:
\begin{equation}\label{eq: open-loop-system}
    x_{k+1} = f(x_k, u_k, d_k)
\end{equation}
The system is controllable with a feedback control law given by $u_k = h(x_k)$. Therefore the closed loop system can be written as:
\begin{equation} \label{eq: closed-loop-system}
    x_{k+1} = F(x_k, d_k)
\end{equation}
First, we provide some definitions required to prove the stability properties.

\begin{definition}[$\mathcal{K}$ function]
A function $\alpha(.) : \mathbb{R} \rightarrow \mathbb{R}$ is a $\mathcal{K}$ function if it is continuous, strictly increasing and satisfies $\alpha(0) = 0$ and $\alpha(s) > 0, \forall s>0$.
\end{definition}

\begin{definition}[$\mathcal{KL}$ function] 
A function $\beta(. , .): \mathbb{R}\times \mathbb{Z} \rightarrow \mathbb{R}$ is a $\mathcal{KL}$ function if $\beta(s, k) $ is a $\mathcal{K}$ function in $s$, for each $s \geq 0$, satisfies $\beta(s, .)$ is non-increasing, and $ \beta(s, k) \rightarrow 0$ as $k \rightarrow \infty$.
\end{definition}

\begin{definition}[$\mathcal{K}_\infty$ function]
A function $\alpha(.) : \mathbb{R} \rightarrow \mathbb{R}$ is a $\mathcal{K}_\infty$ function if it is a $\mathcal{K}$ function and as $s\rightarrow \infty$, $\alpha(s) \rightarrow \infty$.
\end{definition}

%\begin{definition}[Control invariant set] For a dynamic system defined as $x_{k+1} = f(x_k, u_k, d_k)$, the set $\Omega \subseteq \mathbb{X}$ is control invariant if $\forall x \in \Omega, \exists$ a $u \in \mathbb{U}$ such that  $x_{k+1} \in \Omega$ \cite{yu_advanced-step_2019}.
%\end{definition}

\begin{definition}[Robust positive invariant (RPI)] For the system in (\ref{eq: closed-loop-system}), the set $\Omega \subseteq \mathbb{X}$ is robust positive invariant with respect to $\mathbb{D}$ if it contains the origin and if $F(x, d) \in \Omega$, $\forall x \in \Omega$ and $\forall d \in \mathbb{D}$ \cite{rawlings_model_2017}.
\end{definition}

\begin{definition}[ISpS Lyapunov function] A function $V(.)$ is said to be an ISpS Lyapunov function for the system in (\ref{eq: closed-loop-system}), if there exists an RPI set $\Lambda$ containing the origin, a set $\Theta \subseteq \Lambda$ containing the origin, $\mathcal{K}_\infty$ functions $\alpha_1, \alpha_2, \alpha_3$, a $\mathcal{K}$ function $\sigma$, and constants $\tilde{c}_0, \tilde{c}_1 \geq 0$ such that:
\begin{subequations}\label{eq: lyapunov-conditions} 
\begin{equation}
    V(x) \geq \alpha_1(x) \quad \forall x \in \Lambda
\end{equation}
\begin{equation}
    V(x) \leq \alpha_2(x) + \tilde{c}_0 \quad \forall x \in \Theta
\end{equation}
\begin{equation} \label{eq: lyapunov-descent}
    \Delta V(x) = V(F(x, d)) - V(x) \leq -\alpha_3(|x|) + \sigma(|d|) + \tilde{c}_1 \quad \forall x \in \Lambda, \forall d \in \mathbb{D}
\end{equation}
\end{subequations}  
\end{definition}

\begin{definition}[ISpS stability] 
Given a RPI set $\Lambda \in \mathbb{R}^n$ containing the origin, a closed-loop system given by (\ref{eq: closed-loop-system}) is said to be ISpS stable in $\Lambda$ with respect to the origin if there exists a $\mathcal{KL}$ function $\beta$, a $\mathcal{K}$ function $\gamma$ and a constant $\tilde{c}\geq 0$ such that:
\begin{equation}\label{eq: isps-stability}
    |x_k| \leq \beta(|x_0|, k) + \gamma(|d_{[0, k-1]}|) + \tilde{c} \quad \forall x_0 \in \Lambda, \forall k \geq 0, \forall d \in \mathbb{D} 
\end{equation}
where $d_{[0, k-1]}$ indicates the realizations of uncertainty from time 0 to time $k-1$. If equation (\ref{eq: isps-stability}) holds for $\tilde{c} = 0$, then the system (\ref{eq: closed-loop-system}) is input to state stable (ISS) in $\Lambda$.
\end{definition}
A system admitting a Lyapunov function $V$, satisfying equation (\ref{eq: lyapunov-descent}) is ISpS stable \cite{lucia_stability_2020}. Usually the Lyapunov function for an E-NMPC is assumed to be the sum of the tracking objective as shown in (\ref{eq: lyapunov-definition}), given a setpoint. The strict decrease of the Lyapunov function enforced in (\ref{sec:enmpc-formulation}) is enough to ensure the stability of the nominal E-NMPC \cite{griffith_advances_nodate}.
For the multistage E-NMPC, the expected value of the Lyapunov function is enforced to decrease but with a slack to allow the system to be feasible. To show that the multistage E-NMPC with the modified stability constraint is stable, we must show that $\mathbb{E}_\mathbb{C}(V)$ is an ISpS Lyapunov function.

\begin{lemma}[] Given ISpS Lyapunov functions $V^c$ on the RPI set $\Lambda$, $\mathbb{E}_\mathbb{C}(V)$ defined in (\ref{eq: expected-lyapunov}) is an ISpS Lyapunov function in $\Lambda$.
\end{lemma}
\begin{proof}
    Since $V^c$ are ISpS Lyapunov functions, they satisfy equations (\ref{eq: lyapunov-conditions}). Therefore, for $\mathcal{K}_\infty$ functions $\alpha_1^c, \alpha_2^c, \alpha_3^c$ and constants $\tilde{c}_0^c, \tilde{c}_1^c \geq 0, \forall c \in \mathbb{C}$:
     \begin{equation*}
        \mathbb{E}_\mathbb{C}(V(x)) \geq \sum_{c \in \mathbb{C}}p^c \alpha_1^c(x) \quad \forall x \in \Lambda
    \end{equation*}
    Since the weighted average of $\mathcal{K}_\infty$ functions is another $\mathcal{K}_\infty$ function, we can write $\sum_{c \in \mathbb{C}} p^c\alpha_1^c(x) = \alpha_1^{avg}(x)$ where $\alpha_1^{avg}$ is a $\mathcal{K}_\infty$ function. Therefore, 
    \begin{subequations} \label{eq: proof-expected-lyapunov}
        \begin{equation}
        \mathbb{E}_\mathbb{C}(V(x)) \geq \alpha_1^{avg}(x) \quad \forall x \in \Lambda
    \end{equation}
    For the upper-bound on the expected value of the Lyapunov function, given a set a set $\Theta \subseteq \Lambda$ containing the origin, 
    \begin{equation*}
        \mathbb{E}_\mathbb{C}(V(x)) \leq \sum_{c \in \mathbb{C}}p^c (\alpha_2^c(x) + \tilde{c}_0^c) \quad \forall x \in \Theta
    \end{equation*}
    \begin{equation}
        \mathbb{E}_\mathbb{C}(V(x)) \leq \alpha_2^{avg}(x) + \tilde{c}_0^{avg} \quad \forall x \in \Theta
    \end{equation}
    Finally, the descent in the expected value of the Lyapunov function is bounded as follows:
    \begin{equation*}
        \Delta \mathbb{E}_\mathbb{C}(V(x)) = \sum_{c \in \mathbb{C}} p^c(V^c(F(x, d)) - V^c(x))\leq \sum_{c \in \mathbb{C}} p^c(-\alpha_3^c(|x|) + \sigma^c(|d|) + \tilde{c}_1^c) \quad \forall x \in \Lambda, \forall d \in  \mathbb{D}
    \end{equation*}
    \begin{equation}
        \Delta \mathbb{E}_\mathbb{C}(V(x)) \leq -\alpha_3^{avg}(|x|) + \sigma^{avg}(|d|) + \tilde{c}_1^{avg}
    \end{equation}
    From equations (\ref{eq: proof-expected-lyapunov}), it is proved that $\mathbb{E}_\mathbb{C}(V)$ is an ISpS Lyapunov function.
    \end{subequations}
\end{proof}
Therefore, the multistage-ENMPC with the modified stability constraint has input to state practical stability. 

\section{Gas Network Model}
We demonstrate the multistage E-NMPC framework on a gas network model with uncertain parameters.
The gas network model consists of gas transport equations and mass balances across the network. Additionally, it also consists of compressor equations, thermodynamic pressure temperature relationships, and operational constraints. The  model follows from Ghilardi et. al.  \cite{ghilardi_economic_2025} with the main modeling equations listed in Table (\ref{tab:gas-network-equations}) and the nomenclature is given in Table (\ref{tab:gas-network-nomenclature}). The time domain $\mathcal{T}$ is discretized using the backward Euler method and a finite volume method is used to discretize the space domain in the pipes defined by the sets $\epsilon^{in}$ and $\epsilon^{out}$. The system is assumed to be isothermal. 
\begin{table}[h]
    \centering
    \caption{Gas network model}
    \begin{tabular}{c|c}
         Description & Equation \\
         \hline \hline \\
          Mass balance & $\frac{\partial\rho}{\partial t} + \frac{\partial (\rho s)}{\partial l} = 0$ \\ [0.2cm]
          Momentum balance & $\frac{\partial p}{\partial l} +\frac{c_f}{2D}\rho |s|s = 0$\\[0.2cm]
          Equation of state & $p = Z\rho \frac{R}{MW}T$ \\[0.2cm]
         \hline\hline \\
          Nodal mass balance & $\sum_{e \in \epsilon_n^{in}}f_{e,t}^{out} + f_{n,t}^{supply} = \sum_{e \in \epsilon_n^{out}}f_{e,t}^{in} + f_{n,t}^{cons} \quad \forall n \in \mathcal{N}, t \in \mathcal{T}$\\[0.2cm]
          Nodal pressure limits  & $\Tilde{p}_n^{min} \leq p_{n,t} \leq \Tilde{p}_n^{max}  \quad \forall n \in \mathcal{N}, t \in \mathcal{T}$ \\[0.2 cm]
          \hline\hline \\
          Compressor mass balance & $f_{c, t}^{in} = f_{c,t}^{out} \quad \forall c \in \mathcal{C}, t \in \mathcal{T}$\\[0.2cm]
          Compressor boost rule & $\beta_{c,t} = \frac{p_{c,t}^{out}}{p_{c,t}^{in}} \quad \forall c \in \mathcal{C}, t \in \mathcal{T}$ \\[0.2cm]
          Compressor power & $P_{c,t} = \frac{1}{\eta}f^{in}_{c,t}c_pT(\beta_{c,t}^{1-\frac{1}{\gamma}} - 1) \quad \forall c \in \mathcal{C}, t \in \mathcal{T}$ \\[0.2cm]
          Pressure ratio bounds & $\beta_c^{min}\leq\beta_{c,t} \leq \beta_c^{max} \quad \forall c \in \mathcal{C}, t \in \mathcal{T}$      
    \end{tabular}
    \label{tab:gas-network-equations}
\end{table}
\begin{table}[h]
    \centering
    \caption{Nomenclature for the gas network model}
    \begin{tabular}{cc}
         \toprule
         \multicolumn{2}{c}{Sets} \\
         \midrule
         $\mathcal{T}$ & Time \\
         $\mathcal{N}$ & All Nodes in the network \\
         $\mathcal{C}$ & Set of compressors\\
         $\epsilon_n^{in}$ & Set of pipelines with inlet at node $n$\\
         $\epsilon_n^{out}$ & Set of pipelines with outlet at node $n$\\
         \midrule
         \multicolumn{2}{c}{Variables} \\
         \midrule
         $\rho$ & Gas density \\
         $s$ & Gas velocity \\
         $p$ & Gas pressure \\
         $f^{out}_{e,t}$ & Mass flow rate out of edge $e$ at time $t$\\
         $f^{in}_{e,t}$ & Mass flow rate into edge $e$ at time $t$ \\
         $f^{supply}_{n,t}$ & Mass flow rate produced at node $n$ at time $t$ \\
         $f^{in}_{c,t}$ & Mass flow rate entering compressor $c$ at time $t$\\
         $f^{out}_{c,t}$ & Mass flow rate exiting compressor $c$ at time $t$\\
         $p_{n,t}$ & Pressure at node $n$ at time $t$ \\
         $p_{c,t}^{in}$ & Pressure of the gas entering compressor $c$ at time $t$ \\
         $p_{c,t}^{out}$ & Pressure of the gas exiting compressor $c$ at time $t$ \\
         $\beta_{c,t}$ & Compressor pressure ratio in compressor $c$ at time $t$ \\
         \midrule
         \multicolumn{2}{c}{Parameters} \\
         \midrule
         $T$ & Gas temperature \\
         $MW$ & Molecular weight of gas \\
         $Z$ & Gas compressibility factor \\
         $R$ & Ideal gas constant \\
         $f^{cons}_{n,t}$ & Mass flow rate consumed at node $n$ at time $t$ \\
         $\tilde{p}_n^{min}$ & Minimum pressure bound at node $n$\\
         $\tilde{p}_n^{max}$ & Maximum pressure bound at node $n$\\
         $\beta_c^{min}$ & Minimum pressure ratio in compressors \\
         $\beta_c^{max}$ & Maximum pressure ratio in compressors \\
         $\eta$ & Compressor efficiency \\
         $c_p$ & Specific heat capacity of the gas\\          
         \bottomrule
    \end{tabular}
    \label{tab:gas-network-nomenclature}
\end{table}

In this study, the gas consumption $(w^{cons}_{n,t})$ at the sink nodes in the network is assumed to be an uncertain parameter. Compressor power $P$ is the control variable and gas density, speed, pressure, flow rate and compressor pressure ratio $[\rho, s, p, w, \beta]$ are the states. The control objective is to minimize the total energy consumed in the compressors while satisfying the consumer demands. Therefore, the objective function is given by:
\begin{equation}
    l^{ec}_t = \sum_{c \in \mathcal{C}}P_{c,t} \Delta t
\end{equation}
It is important to choose the states and controls on which the terminal constraints are applied in order to not over-specify or under-specify the system (see Example 3.4 in \cite{parker2023dulmage}). The terminal constraint on the states is applied on the intermediate pressures $p$ in the edges $e$. 
\begin{equation} \label{eqn:terminal-state-con}
    p_{e, (N-1)K} = p^{*}_{e, t_0} 
    \quad \forall e \in \epsilon_n^{in}, \forall n \in \mathcal{N}
\end{equation}
The pressures in each pipe at time $t = (N-1)K$ are set equal to the optimal CSS pressures at $t_0$. This constraint ensures that the initial state of the system is the same as the initial optimal CSS. Furthermore, to drive the system to the optimal CSS in the last period, the controls, i.e. the compressor power $P$ at each time point in the last period, are set equal to the optimal CSS controls.
\begin{equation}\label{eqn:terminal-control-con}
    P_{c,t} = P_{c,t}^{*} \quad \forall c \in \mathcal{C}, t\in \mathcal{T}\setminus\{t_0, .. (N-1)K-1\}
\end{equation}
These constraints on the intermediate pressures and compressor power are enough to drive the entire system to optimal CSS in the last period without over-specifying the system. It is worth noting that the control horizon needs to be sufficiently long for the system to be able to reach the terminal region. Strategies to apply terminal constraints when the control horizon is not sufficiently long are mentioned in \cite{wurth_infinite-horizon_2014, durrant-whyte_infinite-horizon_2012}, but are beyond the scope of this paper.

\section{Results}
The multistage E-NMPC formulation is demonstrated on two network topologies, a small test network from \cite{naik_multistage_2023} and a GasLib-40 network from \cite{schmidt_gasliblibrary_2017}. In both problems it is assumed that the demand in the network is cyclic, repeating every 24 hours. The controller horizon is assumed to be 72 hours long with a time discretization of 1 hour.
\subsection{Test network}
\begin{figure}[h]
    \centering
    \includegraphics[width=0.7
    \linewidth]{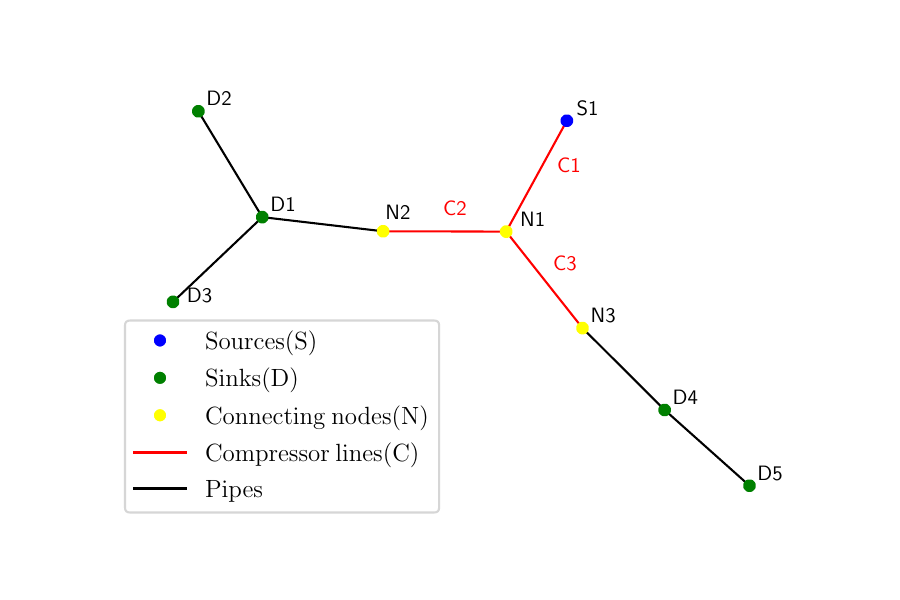}
    \caption{Test network schematic}
    \label{fig:test-network}
\end{figure}
The test network is shown in Figure (\ref{fig:test-network}). It consists of one source node supplying gas at fixed pressure, five sinks that consume gas and three compressor lines that provide boost pressure to the gas. It also consists of three intermediate connecting nodes which do not produce or consume gas. The demand at each sink is assumed to be the same and follows a sinusoidal profile which repeats every 24 hours. The problem is modeled in Pyomo \cite{hart_pyomo_2011} using the MPC extension \cite{parker2023mpc} and solved using the open source nonlinear solver IPOPT \cite{wachter_ipopt_2002,wachter2006ipopt}. We first demonstrate the behavior of a nominal E-NMPC controller in the absence of uncertainty.
\subsubsection{Nominal E-NMPC without any uncertain parameter}
\begin{figure}[h]
\centering
\begin{subfigure}{.5\textwidth}
  \centering
  \includegraphics[width=1\linewidth]{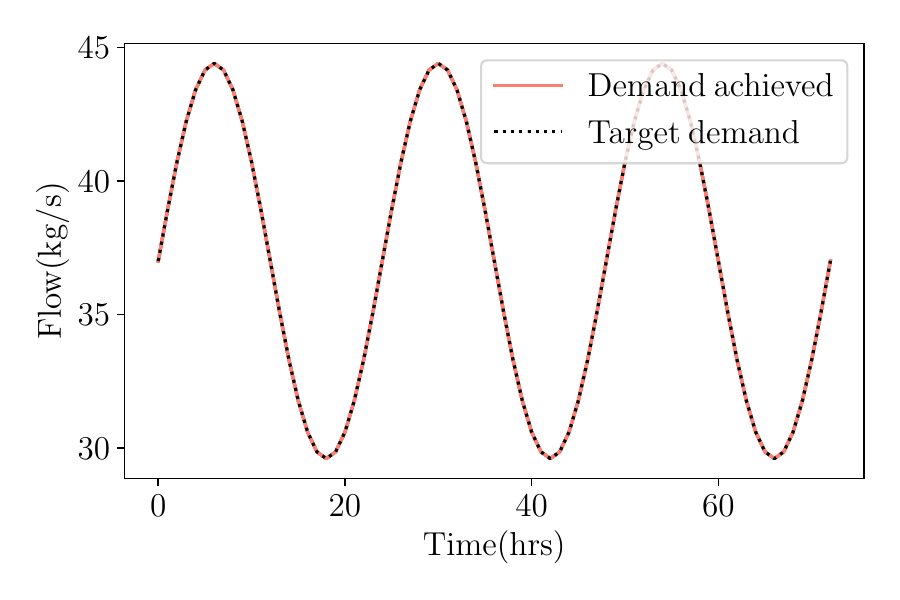}
  \caption{Gas flowrate at sink nodes}
  \label{fig:std-enmpc-no-unc-test-network}
\end{subfigure}%
\begin{subfigure}{.5\textwidth}
  \centering
  \includegraphics[width=1\linewidth]{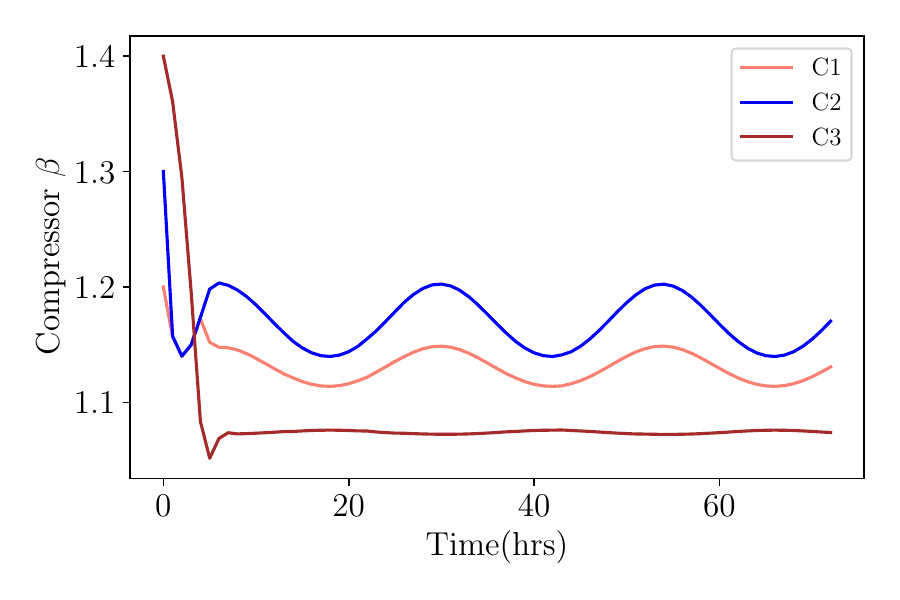}
  \caption{Compressor pressure ratio $\beta = p_{out}/p_{in}$}
  \label{fig:kai-compressor-beta-no-unc}
\end{subfigure}
\caption{Gas flow at sink nodes and the corresponding optimal compressor pressure ratios when there is no uncertain parameter in the test network and the system is controlled using a nominal E-NMPC}
\label{fig:kai-enmpc-no-unc}
\end{figure}
The nominal E-NMPC model consists of 8889 variables and 8763 constraints. Figure (\ref{fig:std-enmpc-no-unc-test-network}) shows that in the absence of uncertainty, a nominal E-NMPC is able to exactly satisfy all sink demands. The corresponding optimal compressor pressure ratios are shown in Figure (\ref{fig:kai-compressor-beta-no-unc}). To minimize energy consumption in the network, the optimal control strategy allows the pressures at the sink nodes $2$ and $5$ to be at their lower bound of 41 bar (Figure (\ref{fig:kai-sink-pressures-no-unc})). The total energy consumed in the network is $7.6 MWh$. 
\begin{figure}[h]
    \centering
    \includegraphics[width=1\linewidth]{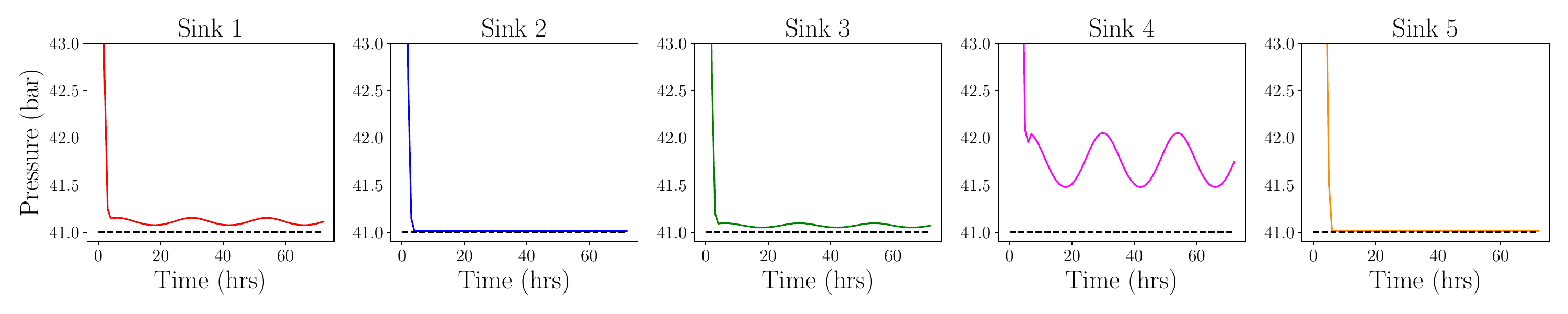}
    \caption{Pressures at sink nodes in the test network, the dashed black line indicates the lower bound}
    \label{fig:kai-sink-pressures-no-unc}
\end{figure}
\begin{figure}[h]
    \centering
    \includegraphics[width=0.5\linewidth]{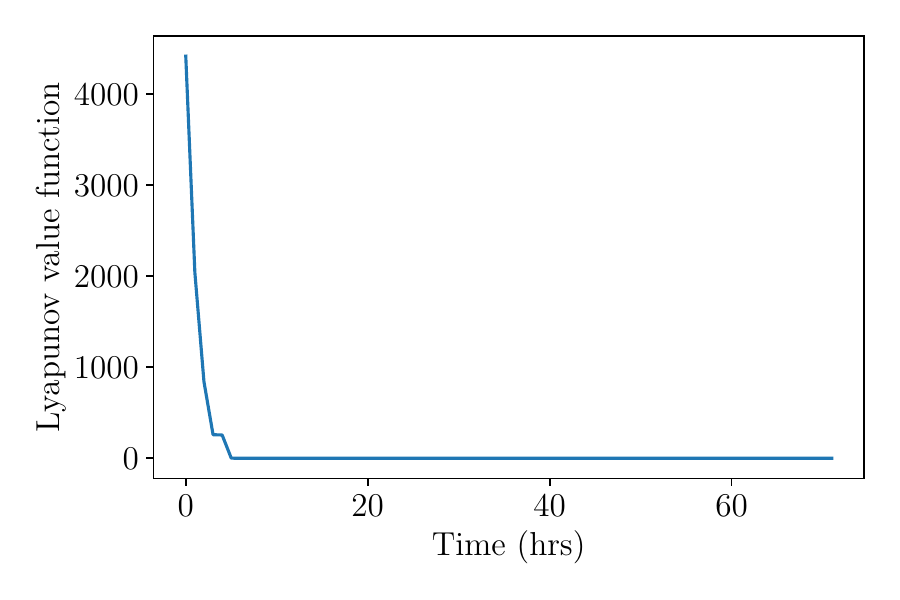}
    \caption{Lyapunov value function when a nominal E-NMPC is used to control the test network in the absence of uncertain parameters in the model}
    \label{fig:kai-lyapunov-func-std-enmpc}
\end{figure}

The nominal E-NMPC is stable, as there are no uncertainties in the model. Figure (\ref{fig:kai-lyapunov-func-std-enmpc}) shows that the Lyapunov function decreases and goes to zero as the system achieves a cyclic steady state. In this simple example, the system goes to cyclic steady state well before the last time period, where the terminal cyclic steady state constraints are enforced.   
\subsubsection{Nominal vs Multistage E-NMPC in presence of uncertain 
parameters}\label{sec:non-vs-multistage-test-network}
\begin{figure}[h]
\centering
\begin{subfigure}{.5\textwidth}
    \centering
    \includegraphics[width=1\linewidth]{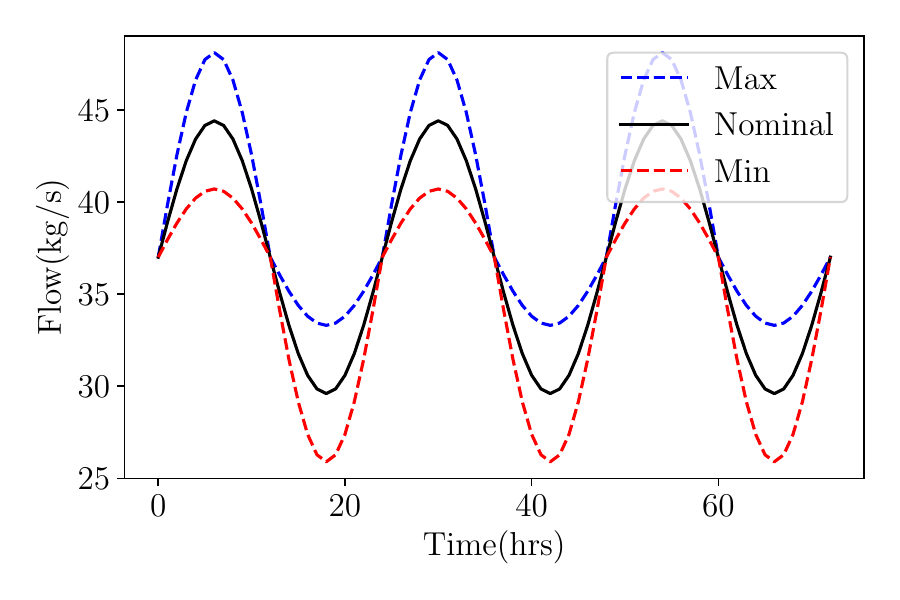}
    \caption{Extreme scenarios}
    \label{fig:kai-extreme-scenarios}
\end{subfigure}%
\begin{subfigure}{.5\textwidth}
  \centering
  \includegraphics[width=1\linewidth]{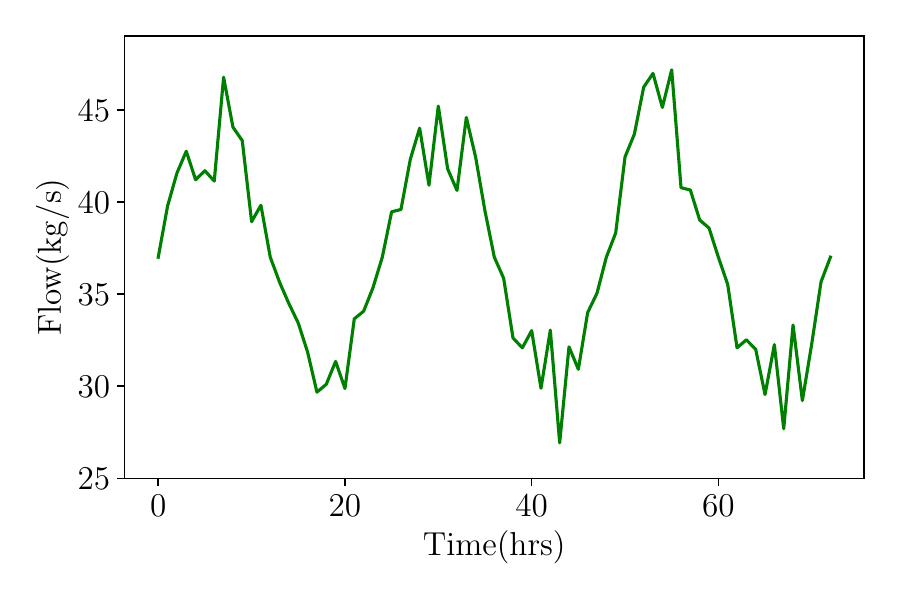}
  \caption{Target demand}
  \label{fig:kai-target-demand}
\end{subfigure}
\caption{Extreme scenarios of demands at sink nodes in the test network and the actual target demand in the test network}
\label{fig:kai-enmpc-no-unc}
\end{figure}
We introduce uncertainty in sink demands in the test network to emulate a more realistic operation. The demand at each sink is allowed to vary between the minimum and maximum demand scenarios shown in Figure (\ref{fig:kai-extreme-scenarios}) while the actual target demand is revealed only in the plant model after the controller computes the robust controls. In this case, the actual target demand, shown in Figure (\ref{fig:kai-target-demand}) is generated by drawing from a uniform distribution.

In this test problem, a nominal E-NMPC meets the sink flow rate requirements exactly using the pipeline gas inventory (linepack). However, it violates pressure bounds at the sinks in the presence of uncertainty in gas flow rates. Figure (\ref{fig:kai-sink-pressure-enmpc-unc}) shows that the lower bound on pressure is violated at sinks 1, 2, 3, and 5. A higher gas flow in the network leads to a greater pressure drop due to friction in the pipes. The nominal E-NMPC does not anticipate the demand uncertainty and therefore violates the pressure limit. The total energy consumed was $7.65 MWh$.
\begin{figure}[h]
    \centering
    \includegraphics[width=1\linewidth]{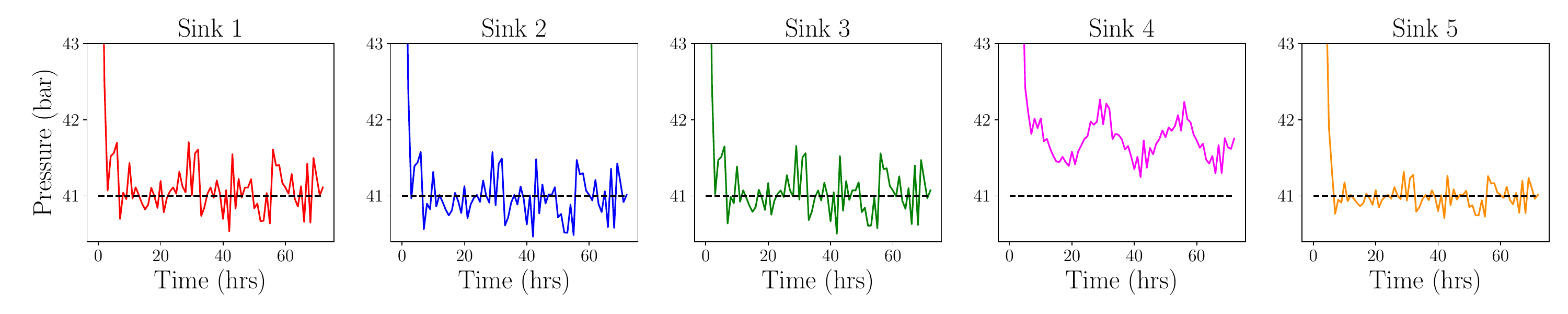}
    \caption{Pressure at sink nodes when a nominal E-NMPC is used to control the test network in presence of demand uncertainty}
    \label{fig:kai-sink-pressure-enmpc-unc}
\end{figure}
\begin{figure}[h]
    \centering
    \includegraphics[width=1\linewidth]{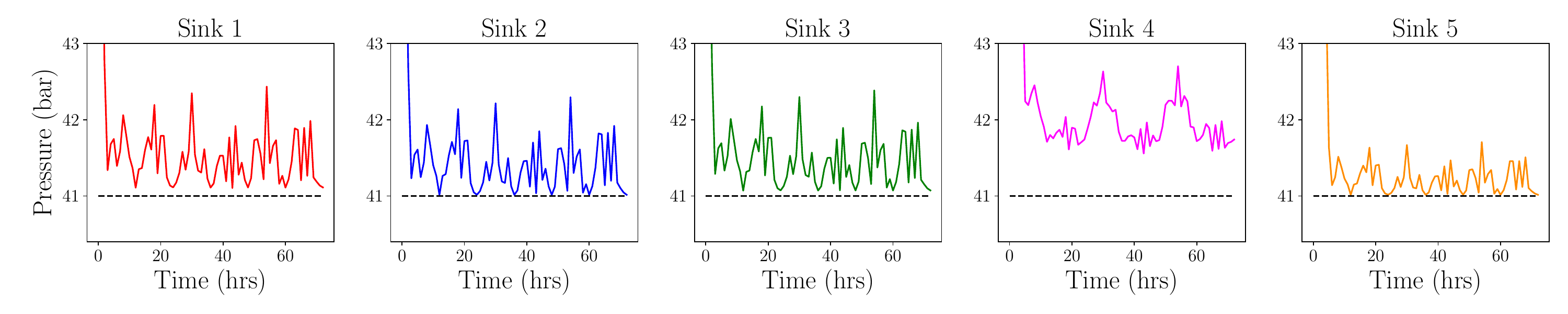}
    \caption{Pressure at sink nodes when a multistage E-NMPC is used to control the test network in presence of demand uncertainty}
    \label{fig:kai-sink-pressures-multistage}
\end{figure}

Instead, a multistage E-NMPC controller is used to control the same system with the three extreme scenarios shown in Figure (\ref{fig:kai-extreme-scenarios}) embedded in the controller. The controller has 26669 variables and 26296 constraints, assuming a robust horizon $H_r = 1$. A multistage E-NMPC satisfies all demands at sink nodes exactly. Additionally, it doesn't violate pressure bounds at the sinks (Figure (\ref{fig:kai-sink-pressures-multistage})), since the extreme scenarios are embedded into the controller. The total energy consumed in the multistage controller is $7.85 MWh$. The slightly higher energy consumption for multistage E-NMPC compared to the nominal E-NMPC is due to the feasibility and robustness of the multistage E-NMPC. The multistage E-NMPC shows input to state practical stability (ISpS). Initially, when the controller is far from the cyclic steady state, the Lyapunov function shows a descent property. However, as we get close to the cyclic steady state, the Lyapunov function has bounded oscillations and does not go to zero due to the uncertainty in the model as shown in Figure (\ref{fig:kai-multistage-lyapunov-stability}). 
\begin{figure}[h]
\centering
\begin{subfigure}{.5\textwidth}
  \centering
  \includegraphics[width=1\linewidth]{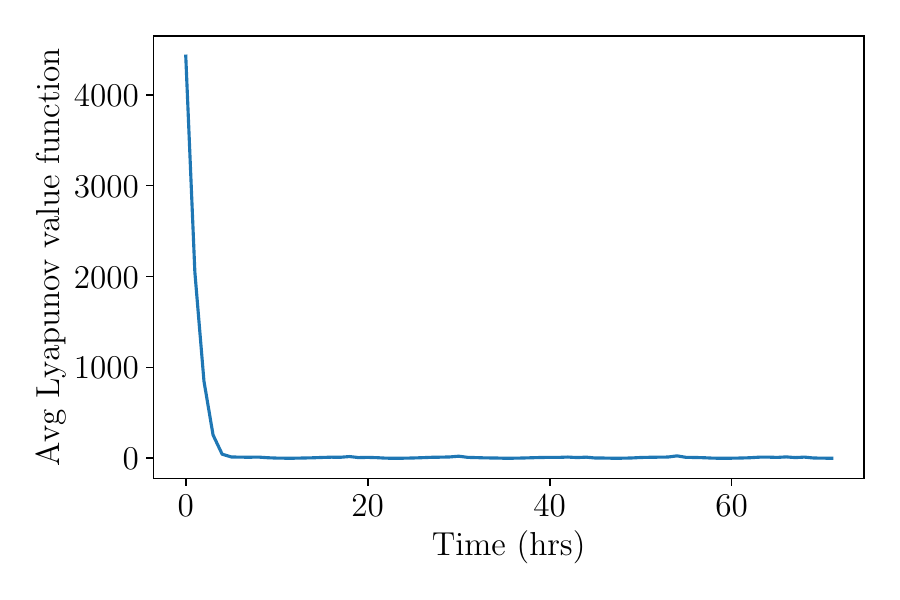}
  \caption{Lyapunov value function}
  \label{fig:kai-multistage-lyapunov-function}
\end{subfigure}%
\begin{subfigure}{.5\textwidth}
  \centering
  \includegraphics[width=1\linewidth]{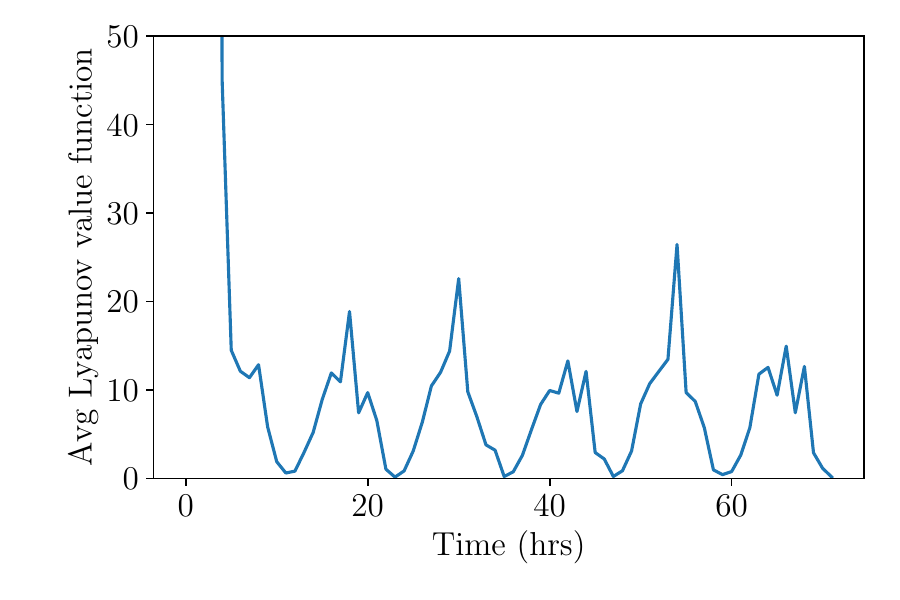}
  \caption{Zoomed in Lyapunov value function of (\ref{fig:kai-multistage-lyapunov-function})}
  \label{fig:kai-multistage-lyapunov-function-zoomed}
\end{subfigure}
\caption{Average Lyapunov value function when a multistage E-NMPC is used to control the test network in presence of demand uncertainty}
\label{fig:kai-multistage-lyapunov-stability}
\end{figure}
\begin{table}[h]
    \centering
    \caption{Results summary for the test network}
    \begin{tabular}{c|ccccccc}
         Controller & Uncertainty & \shortstack{Energy \\ consumption} & \shortstack{Constraint \\ violation} & Variables & Constraints & \shortstack{Average \\ solve-time}\\
         \hline\hline
         Nom E-NMPC & Absent & $7.6 MWh$ & - & 8889 & 8763 & 1.9\\
         Nom E-NMPC & Present & $7.65 MWh$ & Pressure LB & 9984 & 9128 & 2.1 s\\
         MS E-NMPC & Present & $7.85 MWh$ & - & 26669 & 26296 & 7 s
    \end{tabular}
    \label{tab:results-test-network}
\end{table}

The results of the experiments on the test network are summarized in Table (\ref{tab:results-test-network}). The small difference in the number of variables and constraints in the nominal E-NMPC with and without uncertainty is due to the addition of slack variables on the pressure bound at the sinks to retrieve feasibility. In the multistage E-NMPC, the number of variables and constraints are $\approx 3$-fold compared to the nominal E-NMPC. The average solve time of the multistage controller is also $\approx 3.5$ times of the nominal controller. 
\subsection{GasLib-40 network}
\begin{figure}[h]
    \centering
    \includegraphics[width=1\linewidth]{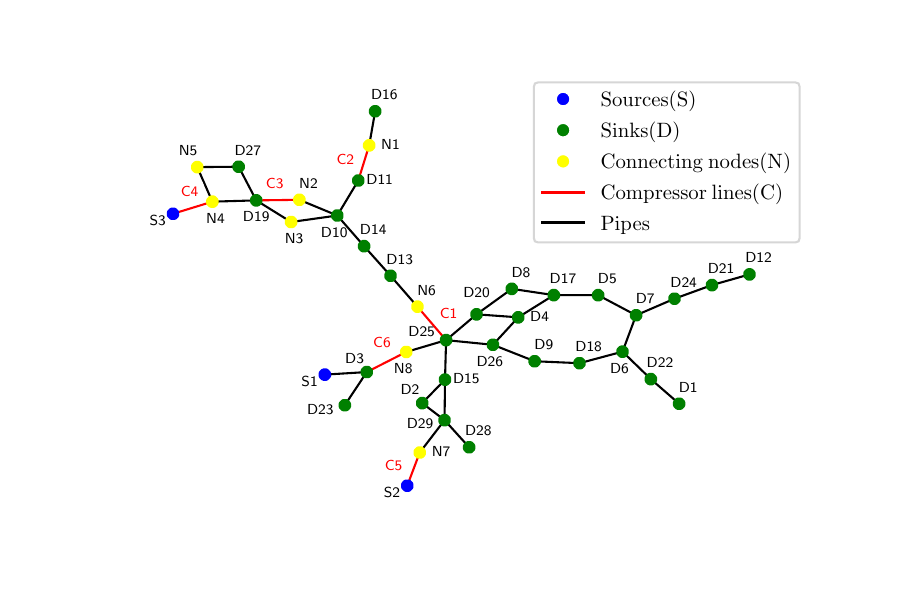}
    \caption{GasLib-40 Network schematic}
    \label{fig:gaslib-40-schematic}
\end{figure}
In this section, we apply the multistage E-NMPC formulation on a larger network with 40 nodes. 
The GasLib-40 network shown in Figure (\ref{fig:gaslib-40-schematic}) consists of 3 sources that supply natural gas either at fixed pressure or fixed flow rate, 29 sinks that consume natural gas, and 8 connecting nodes that do not supply or consume natural gas, simply acting as intermediate nodes in the network. The network has 6 compressors that provide the required boost pressure to the gas, to overcome the pressure loss during gas transport. Sources $S1$ and $S2$ are assumed to be fixed pressure sources and source $S3$ is assumed to have a fixed gas flow rate. Specifying pressure and flow rate both at the source nodes can lead to an ill-defined problem and was studied in detail by Zavala \cite{zavala_stochastic_2014}.
All sink nodes are assumed to have the same sinusoidal cyclic demand profile. 
Additional details about the test problem can be found on the GasLib website: \url{https://gaslib.zib.de/}

\subsubsection{Nominal E-NMPC without any uncertain parameter}
In this section, a nominal E-NMPC is used to optimally control the GasLib-40 network when there is no uncertainty in the model parameters. 
The controller model consists of $65943$ variables and $63661$ constraints. The demand at all the sink nodes is assumed to be cyclic, repeating every 24 hours. The problem is modeled in Pyomo \cite{hart_pyomo_2011} using the MPC extension \cite{parker2023mpc} and solved using the open source nonlinear solver IPOPT \cite{wachter_ipopt_2002,wachter2006ipopt}.

\begin{figure}[h]
\centering
\begin{subfigure}{.5\textwidth}
  \centering
  \includegraphics[width=1\linewidth]{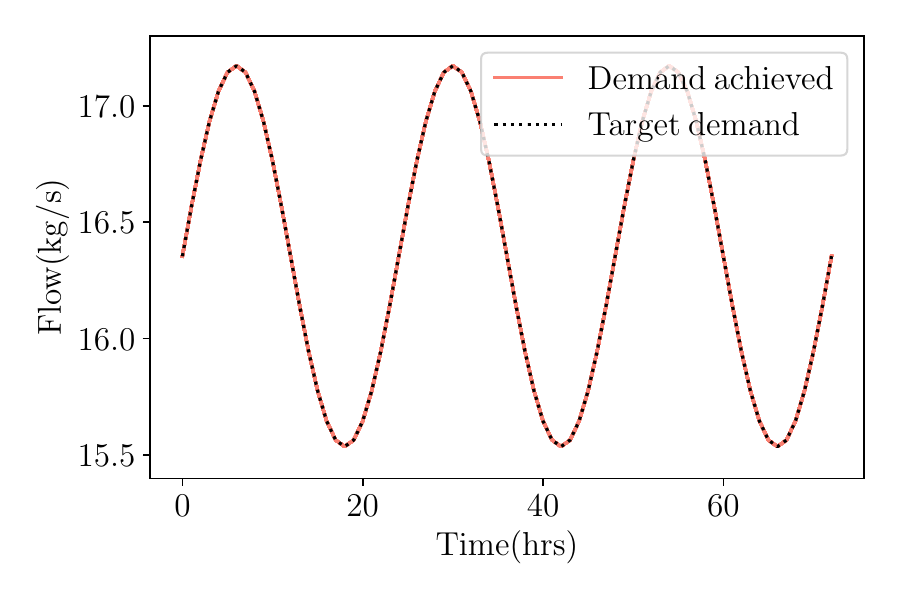}
  \caption{Gas flowrate at sink nodes}
  \label{fig:sink-flow-no-unc}
\end{subfigure}%
\begin{subfigure}{.5\textwidth}
  \centering
  \includegraphics[width=1\linewidth]{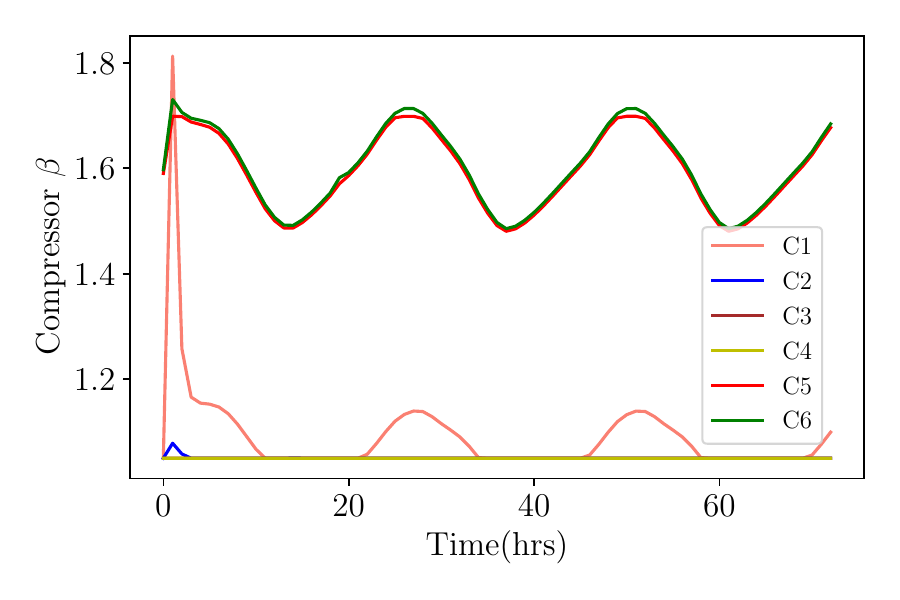}
  \caption{Compressor pressure ratio $\beta = p_{out}/p_{in}$}
  \label{fig:compressor-beta-no-unc}
\end{subfigure}
\caption{Gas demand and the corresponding optimal compressor pressure ratios when there are no uncertain parameters in the network}
\label{fig:enmpc-no-unc}
\end{figure}
Without uncertainty, a nominal E-NMPC is able to satisfy all sink demands exactly as shown in Figure (\ref{fig:sink-flow-no-unc}). Compressors 1, 5 and 6 are active compressors that provide boost pressure to the gas (Figure (\ref{fig:compressor-beta-no-unc})). The total energy consumption to operate the compressors is $2716.45$ \textit{MWh}. The initial increase in the pressure ratio in compressor 1 can be attributed to an initial change in the direction of the gas flow, since the system is initialized at a constant steady-state demand. 

In systems that have a cyclic operation, it is desired to bring the system to the optimal cyclic steady state in the terminal region. The flow rates of the sources $S1$ and $S2$ and the gas pressure of source $S3$ are shown in Figure (\ref{fig:ocss-sources}). Initially, the controller operates away from the cyclic steady state. However, in the last period, between $48-72$ hrs, (beyond the vertical orange line), the system is driven to the optimal cyclic steady state due to the constraints being enforced in the terminal region.
\begin{figure}[h]
    \centering
    \includegraphics[width=1\linewidth]{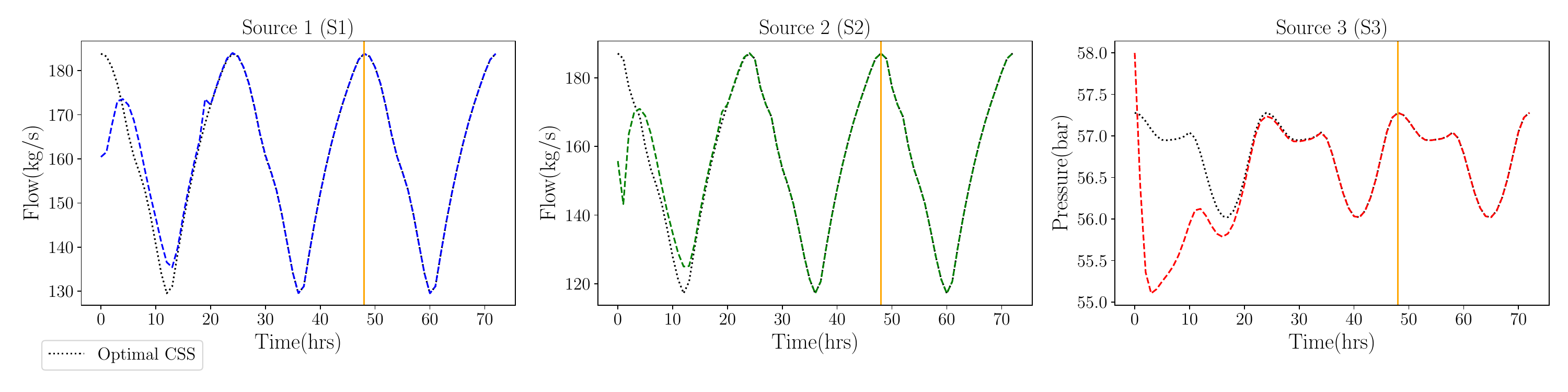}
    \caption{Gas flow rate and gas pressure at the source nodes against the optimal cyclic steady state trajectory. The orange line indicates the start of the last period.}
    \label{fig:ocss-sources}
\end{figure}
The E-NMPC is stable when there are no uncertain parameters in the model. Figure (\ref{fig:lyapunov-value-function-std-enmpc}) shows that the Lyapunov value function has a descent property and eventually converges to the origin indicating the stability of the controller. 

\begin{figure}[h]
    \centering
    \includegraphics[width=0.5\linewidth]{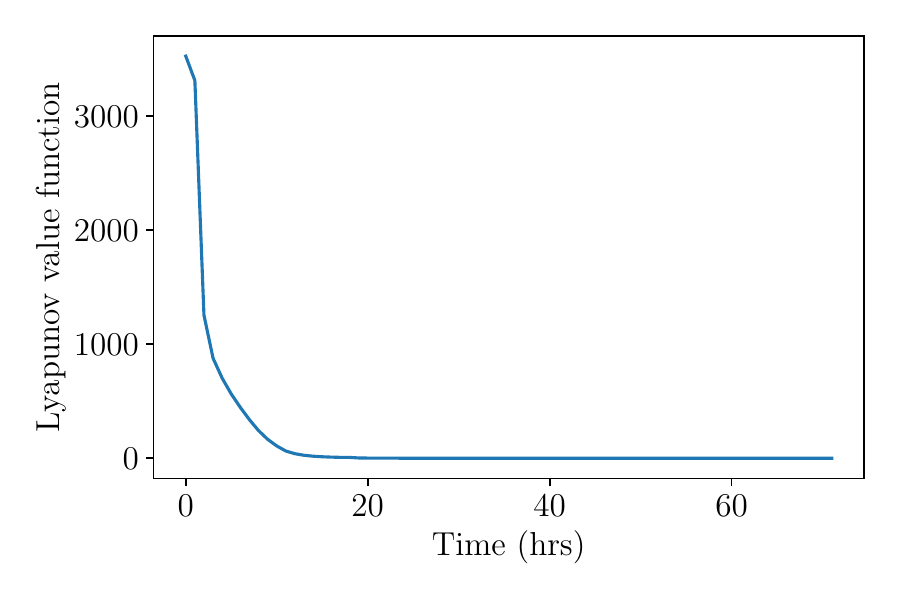}
    \caption{Lyapunov value function for a nominal E-NMPC when there are no uncertain parameters in the network}
    \label{fig:lyapunov-value-function-std-enmpc}
\end{figure}

\subsubsection{Nominal vs Multistage E-NMPC in presence of uncertain 
parameters}\label{sec: std-enmpc-unc-demands}
In this section we assume that the demand at the sink nodes is uncertain, but bounded to be within the extreme scenarios shown in Figure (\ref{fig:gaslib40-enmpc-no-unc}). Using a nominal E-NMPC, when the demands are uncertain, leads to a mismatch in the target demand vs the actual demand that is achieved in the plant model. Figure (\ref{fig:sink-flow-std-enmpc-unc}) shows the mismatch between the target demand and the actual demand achieved in the GasLib-40 network. The demand target is not met exactly for sinks 12 and 21 and is highlighted by the blue rectangle in Figure (\ref{fig:sink-flow-std-enmpc-unc}). The mismatch occurs because the controller is unaware of the uncertainty in the demands and operates assuming nominal demand. The total energy consumed is $2717.16$ \textit{MWh}. 
\begin{figure}[h]
\centering
\begin{subfigure}{.5\textwidth}
    \centering
    \includegraphics[width=1\linewidth]{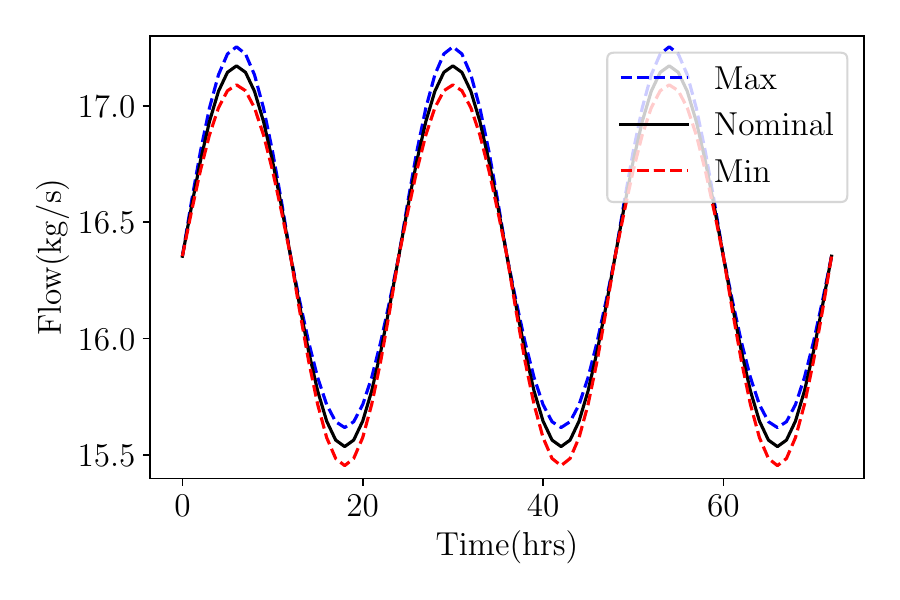}
    \caption{Extreme scenarios}
    \label{fig:gaslib40-extreme-scenarios}
\end{subfigure}%
\begin{subfigure}{.5\textwidth}
  \centering
  \includegraphics[width=1\linewidth]{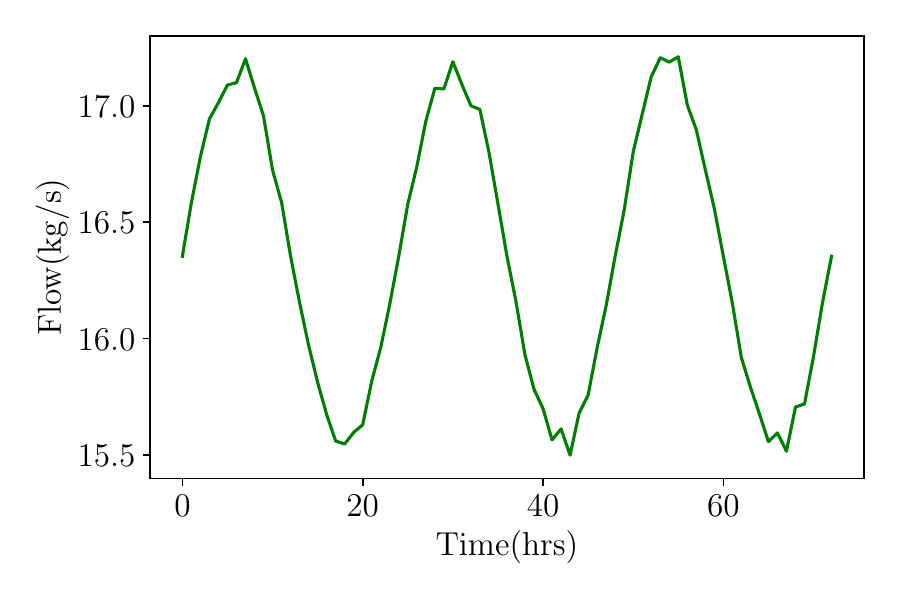}
  \caption{Target demand}
  \label{fig:gaslib-target-demand}
\end{subfigure}
\caption{Extreme scenarios of demands at sink nodes in the GasLib-40 network and the actual target demand in the GasLib-40 network}
\label{fig:gaslib40-enmpc-no-unc}
\end{figure}

\begin{figure}[h]
\centering
\begin{subfigure}{.5\textwidth}
  \centering
  \includegraphics[width=1\linewidth]{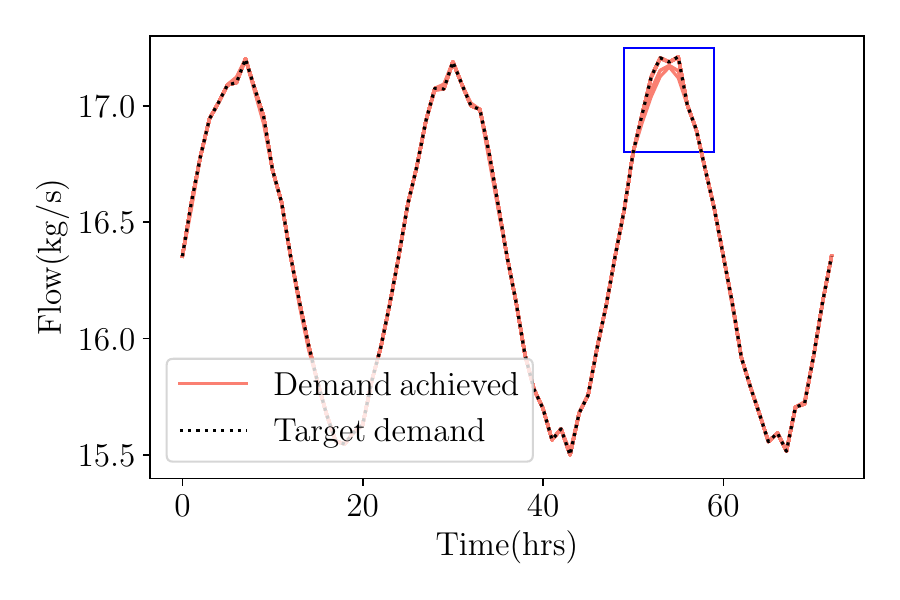}
  \caption{Nominal E-NMPC}
  \label{fig:sink-flow-std-enmpc-unc}
\end{subfigure}%
\begin{subfigure}{.5\textwidth}
  \centering
  \includegraphics[width=1\linewidth]{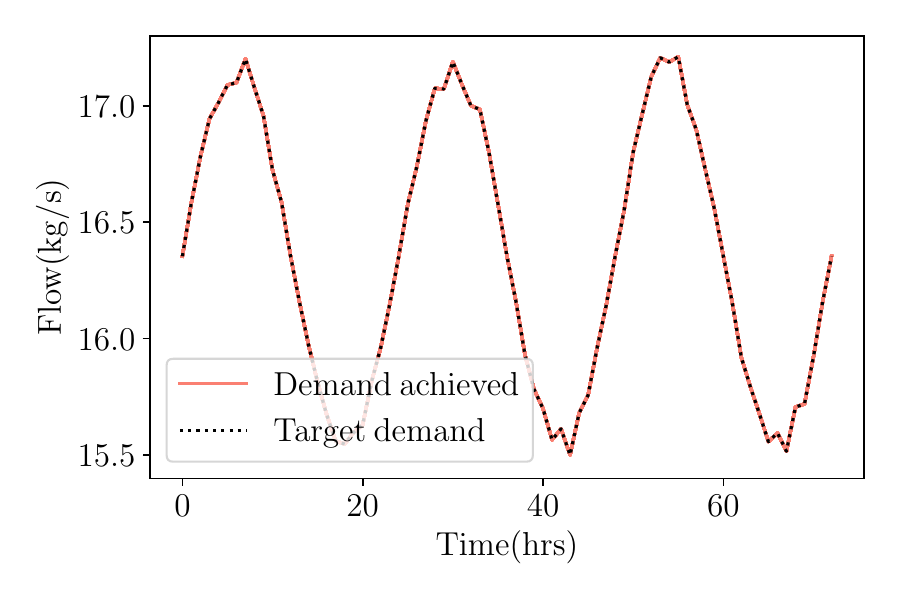}
  \caption{Multistage E-NMPC}
  \label{fig:multistage-enmpc-unc-demands}
\end{subfigure}
\caption{Gas flow rate at sink nodes when the target demand is uncertain and a nominal vs a multistage E-NMPC is used to control the system. The blue rectangle highlights the mismatch in the target and the actual demands achieved.}
\label{fig:enmpc-no-unc}
\end{figure}

Therefore, a multistage E-NMPC controller described in Section (\ref{sec:multistage-enmpc}) is used to control the GasLib-40 network with uncertain demands. Three extreme scenarios shown in Figure (\ref{fig:gaslib40-extreme-scenarios}) are explicitly embedded in the multistage controller and a robust horizon of of 1 time step is considered. The multistage controller model consists of $185129$ variables and $184645$ constraints. The size of the problem increases proportionally to the number of extreme scenarios considered in the multistage controller.

Figure (\ref{fig:multistage-enmpc-unc-demands}) shows that the demands are met exactly at all sink nodes when a multistage E-NMPC is used to control the network in the presence of uncertainties. The total electric energy consumed is $2751$ \textit{MWh}. The slight increase in energy consumption is due to the fact that the multistage E-NMPC considers extreme scenarios while making decisions to avoid constraint violations.

% \begin{figure}[h]
%     \centering
%     \includegraphics[width=0.5\linewidth]{images/compressor_beta_unc_multistage_enmpc.pdf}
%     \caption{Optimal compressor pressure ratios when the target demand is uncertain and a multistage E-NMPC is used to control the system.}
%     \label{fig:multistage-enmc-compressor-beta}
% \end{figure}

\begin{figure}[h]
    \centering
    \includegraphics[width=0.5\linewidth]{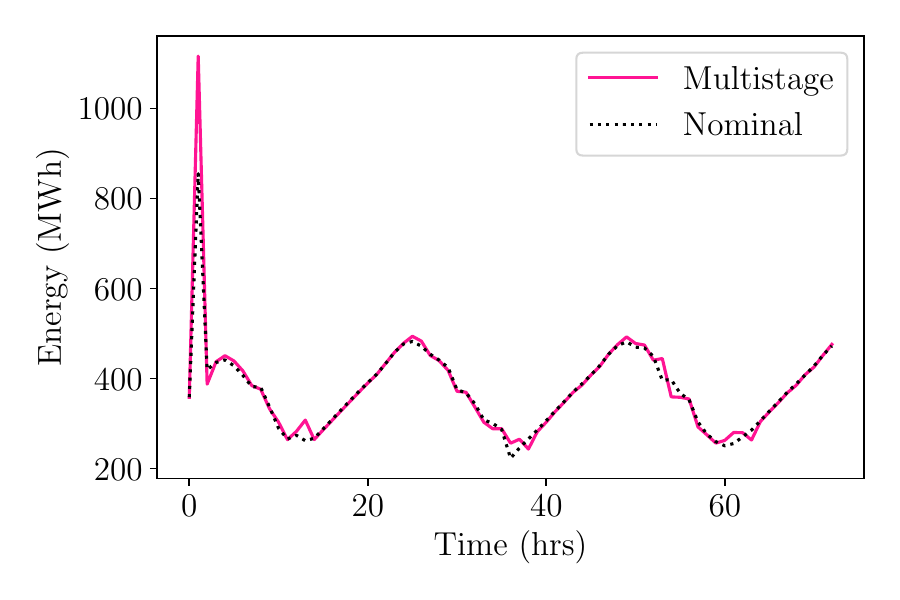}
    \caption{Energy consumption when a nominal E-NMPC is used vs when a multistage E-NMPC is used to control the GasLib-40 network with uncertain demands.}
    \label{fig:multistage-vs-standard}
\end{figure}
Figure (\ref{fig:multistage-vs-standard}) shows that multistage E-NMPC consumes more energy compared to the nominal E-NMPC when the demand is at its peak. The multistage E-NMPC takes conservative steps in order to avoid constraint violations, thus consuming more overall energy. To meet the high flow demands at sinks 12 and 21, the sources 1 and 2 supply higher flow rates of gas when the system is controlled using a multistage E-NMPC.
% \begin{figure}[h]
% \centering
% \begin{subfigure}{.5\textwidth}
%   \centering
%   \includegraphics[width=1\linewidth]{images/source1_flow_multistage_vs_std_enmpc.pdf}
%   \caption{Gas flow rate at source 1 (S1)}
%   \label{fig:source1-ms-vs-std}
% \end{subfigure}%
% \begin{subfigure}{.5\textwidth}
%   \centering
%   \includegraphics[width=1\linewidth]{images/source2_flow_multistage_vs_std_enmpc.pdf}
%   \caption{Gas flow rate at source 2 (S2)}
%   \label{fig:source2-ms-vs-std}
% \end{subfigure}
% \caption{Gas flow rates at sources 1 and 2 when the gaslib-40 network is controlled by a standard E-NMPC vs a multistage E-NMPC}
% \label{fig:source-flows-ms-vs-std}
% \end{figure}

\begin{figure}[h]
\centering
\begin{subfigure}{.5\textwidth}
  \centering
  \includegraphics[width=1\linewidth]{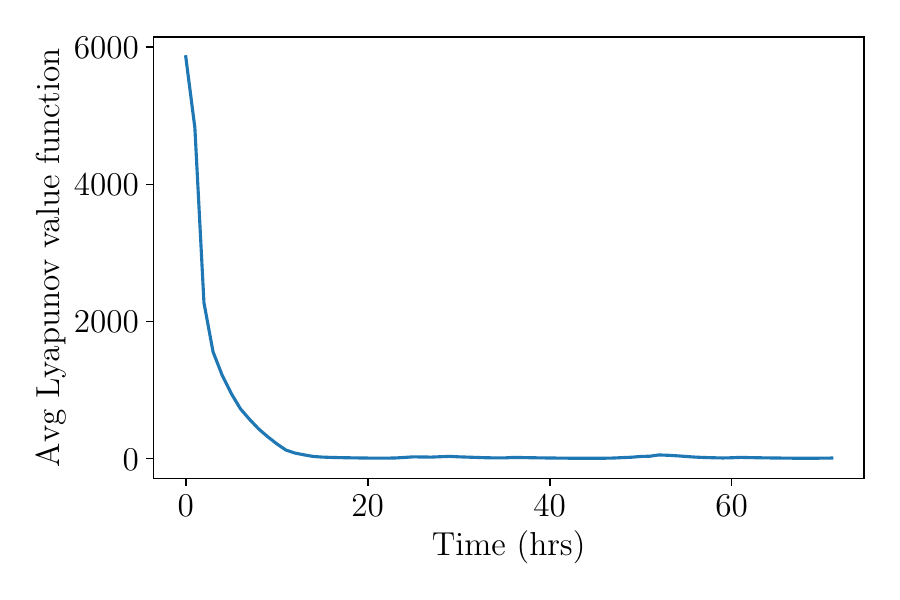}
  \caption{Lyapunov value function}
  \label{fig:multistage-lyapunov-function}
\end{subfigure}%
\begin{subfigure}{.5\textwidth}
  \centering
  \includegraphics[width=1\linewidth]{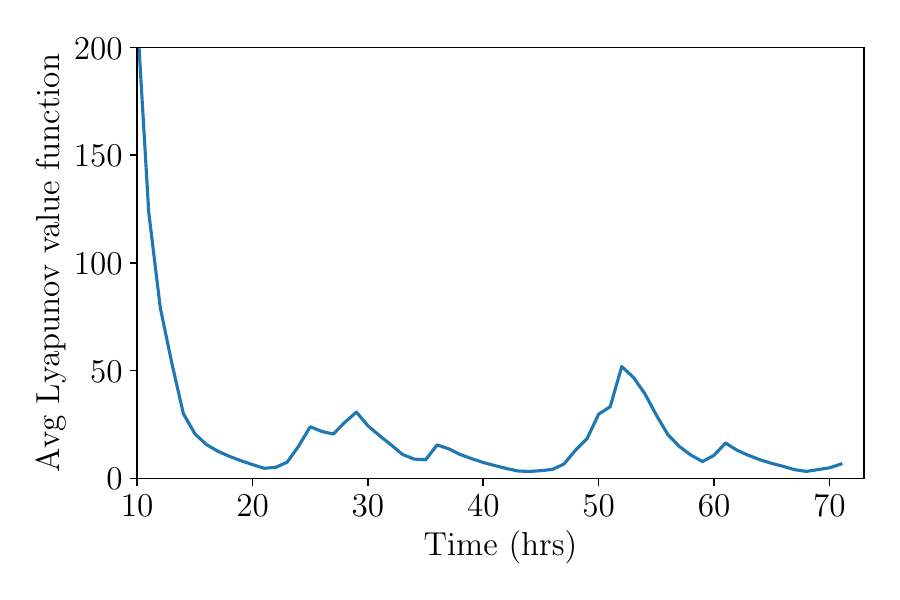}
  \caption{Zoomed in Lyapunov value function of (\ref{fig:multistage-lyapunov-function})}
  \label{fig:multistage-lyapunov-function-zoomed}
\end{subfigure}
\caption{Average Lyapunov value function when a multistage E-NMPC is used to control the GasLib-40 network in presence of demand uncertainty}
\label{fig:multistage-lyapunov-stability}
\end{figure}
The multistage E-NMPC controller has input to state practical stability. Figure (\ref{fig:multistage-lyapunov-function}) shows the expected value of the Lyapunov value function for the multistage controller has an initial descent property. However, as we move closer to the origin, Figure (\ref{fig:multistage-lyapunov-function-zoomed}) shows that the expected value of the Lyapunov value function has bounded oscillations. Due to the uncertain parameter, it is not possible to drive the system exactly to the optimal cyclic steady state. The results of the GasLib-40 network are summarized in Table \ref{tab:results-gaslib40-network}. The multistage E-NMPC takes $\approx 8$ times more time to solve than the nominal E-NMPC with uncertainty.
\begin{table}[h]
    \centering
    \caption{Results summary for the GasLib-40 network}
    \begin{tabular}{c|ccccccc}
         Controller & Uncertainty & \shortstack{Energy \\ consumption} & \shortstack{Constraint \\ violation} & Variables & Constraints & \shortstack{Average \\ solve-time}\\
         \hline\hline
         Nom E-NMPC & Absent & $2716 MWh$ & - & 61709 & 61544 & 21 s\\
         Nom E-NMPC & Present & $2717 MWh$ & Sink demands & 65943 & 63661 & 41 s\\
         MS E-NMPC & Present & $2751 MWh$ & - & 185129 & 184645 & 336 s
    \end{tabular}
    \label{tab:results-gaslib40-network}
\end{table}
\section{Conclusions}
In this paper, a multistage E-NMPC formulation is developed. Endpoint constraints are incorporated to account for the cyclic steady state of the controlled system. A Lyapunov based stability constraint is enforced to ensure the multistage E-NMPC has the ISpS property. We test the multistage E-NMPC formulation on two gas network operation problems with uncertain demands. Our results show that the multistage E-NMPC is robust to uncertainties that lie within the extreme cases embedded in the controller. The experiments also demonstrate that the controller is robustly stable, and oscillates close to its optimal cyclic steady state due to the presence of uncertain parameters. 

In the small gas network, the uncertainty in customer demands leads to a violation in pressure bound when the system is controlled using a nominal E-NMPC. The multistage E-NMPC makes conservative decisions and effectively prevents the pressure bound violation in presence of uncertain sink demands. In the GasLib-40 network, with the uncertainty in customer demands, the nominal E-NMPC does not violate pressure bounds but it instead does not respect the flow demands at nodes that are far from the suppliers. Conversely, a multistage controller can prevent the mismatch between target demand and actual gas delivered, as it considers the extreme cases in uncertain scenarios. In the test network, the multistage E-NMPC requires $3.2 \%$ more energy compared to the nominal E-NMPC while in the GasLib-40 network, the  increase in energy consumption is only $1.2 \%$, in exchange for more robust solutions. 

The stability and terminal constraints in the current form require pre-computation of the optimal cyclic steady state which can change frequently in the presence of uncertainties. Future work would focus on stabilization strategies that eliminate the need for a precomputed cyclic steady state \cite{lin_self-stabilizing_2023}.
Additionally, despite the presence of robust horizon, the multistage E-NMPC framework becomes intractable with an increase in the number of uncertain parameters and the extreme cases of each parameter. This challenge has been overcome in recent work that focuses on scenario selection through sensitivity analysis to decrease the size of the scenario tree and accelerate computations \cite{thombre_sensitivity-assisted_2021, suwartadi_sensitivity-based_2017}. The extension of the multi-stage E-NMPC framework with scenario selection also needs to be explored.
Finally, while solve times here are relatively modest compared to the 1-hour time step, with more uncertain parameters it may become necessary to accelerate solution of the control problem \cite{word2014parallel}.

\section*{Acknowledgements}
RP acknowledges financial support from the Center for Nonlinear Studies at
Los Alamos National Laboratory. Approved for unlimited release. LA-UR-25-22063.

\newpage
\bibliography{references}

\end{document}